%% file: main.tex
\documentclass[submission, copyright, creativecommons]{eptcs}

\usepackage{amssymb}
\usepackage{amsmath}
\usepackage{centernot}
\usepackage{stmaryrd}
\usepackage{xspace}
\usepackage{listings}
\usepackage{hyperref}
\usepackage{macros}
\usepackage{pifont}
\usepackage[dvipsnames]{xcolor}
\usepackage{tikz-cd}
\usepackage{graphicx}
\usepackage{array}
\usepackage{todonotes}

\newif\ifsubmit
\submittrue
\ifsubmit
 
\newcommand{\DAComm}[1]{} 
\else
 
\newcommand{\DAComm}[1]{{\scriptsize\textcolor{magenta}{[\bf{Davide: }#1}]}}
\fi

\lstdefinelanguage{rml}{
  morekeywords={not,false,true,none,any,empty,matches,all,let,if,else},
  morecomment=[l]{//},
}

\title{Can determinism and compositionality coexist in \rml?}

\author{Davide Ancona \qquad
Viviana Mascardi
\institute{DIBRIS, University of Genova, Italy}
\email{\{Davide.Ancona,Viviana.Mascardi\}@unige.it}
\and
Angelo Ferrando
\institute{University of Manchester, UK}
\email{angelo.ferrando@manchester.ac.uk}
}
\def\titlerunning{Can determinism and compositionality coexist in \rml?}
\def\authorrunning{D. Ancona, A. Ferrando, and V. Mascardi}

\allowdisplaybreaks
\begin{document}

\maketitle

\input{abstract}
\input{intro}
\input{background}
\input{calculus}
\input{comp}
\input{related}
\input{conclu}
\bibliographystyle{eptcs}
\bibliography{main}
\end{document}

%% file: abstract.tex
\begin{abstract}
Runtime verification (RV) consists in dynamically verifying that
the event traces generated by single runs of a system under scrutiny (SUS)
are compliant with the formal specification of its expected properties.
\rml (Runtime Monitoring Language) is a simple but expressive
Domain Specific Language for RV; its semantics is based on a trace calculus formalized by a deterministic rewriting  system
which drives the implementation of the interpreter of the monitors
generated by the \rml compiler from the specifications.
While determinism of the trace calculus ensures better performances of the
generated monitors, it makes the semantics of its operators less intuitive.
In this paper we move a first step towards a compositional semantics of
the \rml trace calculus, by interpreting its basic operators 
as operations on sets of instantiated event traces and by proving that such
an interpretation is equivalent to the operational semantics of the calculus. 
\end{abstract}

%% file: intro.tex
\section{Introduction}
RV \cite{rv,FalconeKRT18,BartocciFFR18} consists in dynamically verifying that the event traces generated by single runs of a SUS
are compliant with the formal specification of its expected properties.

The RV process needs as inputs the SUS and the specification of the properties to be verified,
usually defined with either a domain specific (DSL) or a programming language,
to denote the set of valid event traces; RV is performed by monitors, automatically generated from
the specification, which consume the observed events of the SUS, emit verdicts and, in case they work online
while the SUS is executing, feedback useful for error recovery.

RV is complimentary to other verification methods:
analogously to formal verification, it uses a specification formalism, but, as opposite to it, scales well to real systems and complex properties
and it is not exhaustive as happens in software testing;
however, it also exhibits several distinguishing features: it is quite useful to check control-oriented properties \cite{AhrendtCPS17},
and offers opportunities for fault protection when the monitor runs online. Many RV approaches adopt a DSL language
to specifiy properties to favor portability and reuse of specifications and interoperability of the generated monitors and to provide
stronger correctness guarantees: monitors automatically generated from a higher level DSL are more reliable than
ad hoc code implemented in a ordinary programming language to perform RV.

\rml\footnote{\url{https://rmlatdibris.github.io}} \cite{FranceschiniPhD2020} is a simple but expressive DSL for RV
which can be used in practice for RV of complex non Context-Free properties, as FIFO properties, which can be verified
by the generated monitors in time linear in the size of the inspected trace; the language design and implementation is based on previous work on trace expressions and global types \cite{AnconaDM12,CastagnaEtAl12,AnconaBB0CDGGGH16,AnconaFM17}, which have been adopted for RV in several contexts.
Its semantics is based on a trace calculus formalized by a rewriting  system
which drives the implementation of the interpreter of the monitors generated by the \rml compiler from the specifications;
to allow better performances, the rewriting system is fully deterministic \cite{AnconaFFM19} \added{by adopting a left-preferential evaluation strategy
for binary operators} and, thus, no monitor backtracking is needed and
exponential explosion of the space allocated for the states of the monitor is avoided. \added{A similar strategy is followed by 
  mainstream programming languages in predefined libraries for regular expressions for efficient incremental matching of input sequences,
  to avoid the issue of Regular expression Denial of Service (ReDoS) \cite{DavisCSL18}: for instance,
  given the regular expression \lstinline{a?(ab)?} (optionally \lstinline{a} concatenated with optionally \lstinline{ab}) and the input sequence
  \lstinline{ab}, the Java method \lstinline{lookingAt()} of class \lstinline{java.util.regex.Matcher} matches \lstinline{a} instead of the entire input
  sequence \lstinline{ab} because the evaluation of concatenation is deterministically left-preferential.}
  
\added{As explained more in details in Section~\ref{sect:related_work}, with respect to other existing RV formalisms,
  \rml has been designed as an extension of regular expressions and deterministic context-free grammars, which are widely used in RV
  because they are well-understood among software developers as opposite to other more sophisticated approaches, as temporal logics.
As shown in previous papers \cite{AnconaDM12,CastagnaEtAl12,AnconaBB0CDGGGH16,AnconaFM17}, the calculus at the basis of \rml 
allows users to define and efficiently check complex parameterized properties
and it has been proved to be more expressive than LTL \cite{AnconaFM16}}.

Unfortunately, while determinism ensures better performances,
it makes the \added{compositional} semantics of its operators less intuitive; \added{for instance, the example above concerning the
  regular expression  \lstinline{a?(ab)?} with deterministic left-preferential concatenation applies also to \rml, which is more expressive than regular
  expressions: the compositional semantics of concatenation does not correspond to standard language concatenation, because 
  \lstinline{a?} and \lstinline{(ab)?} denote the formal languages $\{\emptyevtr,a\}$ and $\{\emptyevtr,ab\}$, respectively, where $\emptyevtr$ denotes
  the empty string, while, if concatenation is deterministically left-preferential, then the semantics of \lstinline{a?(ab)?} is $\{\emptyevtr,a,aab\}$ which
  does not coincide with the language $\{\emptyevtr,a,ab,aab\}$ obtained by concatenating $\{\emptyevtr,a\}$ with $\{\emptyevtr,ab\}$. In Section~\ref{sec:comp} we show that the semantics of left-preferential concatenation can still be given compositionally, although the corresponding operator is more complicate than standard language concatenation. Similar results follow for the other binary operators of \rml (union, intersection and shuffle);
  in particular, the compositional semantics of left-preferential shuffle is more challenging. Furthermore, the fact that \rml supports parametricity
makes the compositional semantics more complex, since traces must be coupled with the corresponding substitutions generated by event matching. To this aim, as a first step towards a compositional semantics of
the \rml trace calculus,} we provide an interpretation of the basic operators of the \rml trace calculus 
as operations on sets of instantiated event traces, that is, pairs of trace of events and
substitutions computed to bind the variables occurring in the event type patterns used in the specifications and to associate
them with the data values carried by the matched events.
\added{Furthermore} we prove that such
an interpretation is equivalent to the original operational semantics of the calculus based on the deterministic rewriting  system.

The paper is structured as follows: Section~\ref{sec:back} introduces the basic definitions which are used in the subsequent technical sections,
Section~\ref{sec:calculus} formalizes the \rml trace calculus and its operational semantics, while Section~\ref{sec:comp} introduces the
semantics based on sets of instantiated event traces and formally proves its equivalence with the operational semantics; finally,
\added{Section~\ref{sect:related_work} is devoted to the related work and} Section~\ref{sec:conclu}
draws conclusions and directions for further work.
For space limitations, some proof details can be found in the extended version \cite{ancona2020determinism} of this paper.

%% file: background.tex
\section{Technical background}\label{sec:back}
This section introduces some basic definitions and propositions used in the next technical sections.
\paragraph{Partial functions:} Let $\fun{f}{D}{C}$ be a partial function; then $\dom(f)\subseteq D$ denotes the set
of elements $d\in D$ s.t. $f(d)$ is defined (hence, $f(d)\in C$).

A partial function over natural numbers $\fun{f}{\nat}{N}$, with $N\subseteq\nat$,  is \emph{strictly increasing} iff
for all $n_1,n_2\in\dom(f)$, $n_1<n_2$ implies $f(n_1)<f(n_2)$. From this definition one can easily deduce that
a strictly increasing partial function over natural numbers is always injective, and, hence, it is bijective iff it is surjective. 

\begin{proposition}\label{prop:totality}
Let $\fun{f}{\nat}{N}$, with $N\subseteq\nat$, be a strictly
increasing partial function. Then for all $n_1,n_2\in\dom(f)$, if $f(n_1)<f(n_2)$, then
$n_1<n_2$.
\end{proposition}

\begin{proposition}\label{prop:identity}
  Let $\fun{f}{\nat}{N}$, with $N\subseteq\nat$, be a strictly
  increasing partial function satisfying the following conditions:
  \begin{enumerate}
  \item $f$ is surjective (hence, bijective);
  \item for all $n\in\nat$, if $n+1\in\dom(f)$, then $n\in\dom(f)$;
  \item for all $n\in\nat$, if $n+1\in N$, then $n\in N$;
  \end{enumerate}
  Then, for all $n\in\nat$, if $n\in\dom(f)$, then $f(n)=n$, hence $f$ is the identity over $\dom(f)$, and $\dom(f)=N$.
\end{proposition}
  
\paragraph{Event traces:}
Let $\eventSet$ denotes a possibly infinite set $\eventSet$ of events, called the \emph{event universe}.
An event trace over the event universe $\eventSet$ is a partial function $\fun{\evtr}{\nat}{\eventSet}$ s.t.
for all $n\in\nat$, if $n+1\in\dom(\evtr)$, then $n\in\dom(\evtr)$. We call $\evtr$ \emph{finite}/\emph{infinite}
iff $\dom(\evtr)$ is finite/infinite, respectively; when $\evtr$ is finite, its length $\lgth{\evtr}$ coincides with the cardinality
of $\dom(\evtr)$, while $\lgth{\evtr}$ is undefined for infinite traces $\evtr$.
From the definitions above one can easily deduce that if $\evtr$ is finite, then $\dom(\evtr)=\{n\in\nat \mid n < \lgth{\evtr}\}$.
We denote with
$\emptyevtr$ the unique trace over $\eventSet$ s.t. $\lgth{\emptyevtr}=0$; when not ambiguous, we denote with $\ev$ the
trace $\evtr$ s.t. $\lgth{\evtr}=1$ and $\evtr(0)=\ev$.

For simplicity, in the rest of the paper we implicitly assume that all considered event traces are defined over the same event universe.

\paragraph{Concatenation:}
The concatenation $\evtr_1\catop\evtr_2$ of event trace $\evtr_1$ and $\evtr_2$  is the trace $\evtr$ s.t.
\begin{itemize}
\item if $\evtr_1$ is infinite, then $\evtr=\evtr_1$; 
\item if $\evtr_1$ is finite, then $\evtr(n)=\evtr_1(n)$ for all $n\in\dom(\evtr_1)$,
  $\evtr(n+\lgth{\evtr_1})=\evtr_2(n)$ for all $n\in\dom(\evtr_2)$, and if $\evtr_2$ is finite, then
  $\dom(\evtr)=\{n\mid n < \lgth{\evtr_1}+\lgth{\evtr_2}\}$.
\end{itemize}
From the definition above one can easily deduce that $\emptyevtr$ is the identity of $\catop$, and that
$\evtr_1\catop\evtr_2$ is infinite iff $\evtr_1$ or $\evtr_2$ is infinite.
The trace $\evtr_1$ is a prefix of $\evtr_2$, denoted with $\evtr_1\pref\evtr_2$, iff there exists $\evtr$ s.t. $\evtr_1\catop\evtr=\evtr_2$.
If $\evtrSet_1$ and $\evtrSet_2$ are two sets of event traces over $\eventSet$, then $\evtrSet_1\catop\evtrSet_2$ is the
set $\{ \evtr_1\catop\evtr_2 \mid \evtr_1\in\evtrSet_1, \evtr_2\in\evtrSet_2\}$. We write $\evtr_1\pref\evtrSet$ to mean that
there exists $\evtr_2\in\evtrSet$ s.t. $\evtr_1\pref\evtr_2$.

\paragraph{Shuffle:}
The shuffle $\evtr_1\shuffleop\evtr_2$ of event trace $\evtr_1$ and $\evtr_2$ is the set of traces $\evtrSet$ s.t. $\evtr\in\evtrSet$ iff
$\dom(\evtr)$ can be partitioned into $N_1$ and $N_2$ in such a way that there exist two strictly increasing and bijective\footnote{Actually, the sufficient condition is surjectivity, but bijectivity can be derived from the fact that the functions are strictly increasing over natural numbers.} partial functions
$\fun{f_1}{\dom(\evtr_1)}{N_1}$ and $\fun{f_2}{\dom(\evtr_2)}{N_2}$ s.t.
\begin{itemize}
\item[] $\evtr_1(n_1)=\evtr(f_1(n_1))$ and $\evtr_2(n_2)=\evtr(f_2(n_2))$, for all $n_1\in\dom(\evtr_1)$, $n_2\in\dom(\evtr_2)$.
\end{itemize}  
From the definition above, the definition of $\emptyevtr$ and Proposition~\ref{prop:identity} one can deduce that
$\emptyevtr\shuffleop\evtr=\evtr\shuffleop\emptyevtr=\{\evtr\}$; it is easy to show that for all $\evtr\in\evtr_1\shuffleop\evtr_2$,
$\evtr$ is infinite iff $\evtr_1$ or $\evtr_2$ is infinite, and $\lgth{\evtr}=n$ iff $\lgth{\evtr_1}=n_1$, $\lgth{\evtr_2}=n_2$ and $n=n_1+n_2$. 

If $\evtrSet_1$ and $\evtrSet_2$ are two sets of event traces over $\eventSet$, then $\evtrSet_1\shuffleop\evtrSet_2$ is the
set $\bigcup_{\evtr_1\in\evtrSet_1,\evtr_2\in\evtrSet_2}(\evtr_1\shuffleop\evtr_2)$.

\paragraph{Left-preferential shuffle:}
The left-preferential shuffle $\evtr_1\lpshuffleop\evtr_2$ of event trace $\evtr_1$ and $\evtr_2$ is the set of traces
$\evtrSet\subseteq\evtr_1\shuffleop\evtr_2$ s.t. $\evtr\in\evtrSet$ iff
$\dom(\evtr)$ can be partitioned into $N_1$ and $N_2$ in such a way that there exist two strictly increasing and bijective partial functions
$\fun{f_1}{\dom(\evtr_1)}{N_1}$ and $\fun{f_2}{\dom(\evtr_2)}{N_2}$ s.t.
\begin{itemize}
\item $\evtr_1(n_1)=\evtr(f_1(n_1))$ and $\evtr_2(n_2)=\evtr(f_2(n_2))$, for all $n_1\in\dom(\evtr_1)$, $n_2\in\dom(\evtr_2)$;  
\item for all $n_2\in\dom(\evtr_2)$, if $m=\min\{ n_1\in\dom(\evtr_1) \mid f_2(n_2)<f_1(n_1) \}$, then $\evtr_1(m)\neq\evtr_2(n_2)$.  
\end{itemize}  
In the definition above, if\footnote{This happens iff in $\evtr$ all events of $\evtr_1$ precede position $n_2$, hence, event $\evtr_2(n_2)$.} $\{ n_1\in\dom(\evtr_1) \mid f_2(n_2)<f_1(n_1) \}=\emptyset$, then the second condition trivially holds.

\added{As an example, if we have two traces of events $\evtr_1 = \ev_1\catop\ev_2$, and $\evtr_2 = \ev_2\catop\ev_3$, by applying the left-preferential shuffle we obtain the set of traces $\evtr_1\lpshuffleop\evtr_2 = \{ \ev_1\catop\ev_2\catop\ev_2\catop\ev_3, \ev_2\catop\ev_3\catop\ev_1\catop\ev_2, \ev_2\catop\ev_1\catop\ev_3\catop\ev_2, \ev_2\catop\ev_1\catop\ev_2\catop\ev_3 \}$. With respect to $\evtr_1\shuffleop\evtr_2$, the trace $\ev_1\catop\ev_2\catop\ev_3\catop\ev_2$ has been excluded, since this can be obtained only when the first occurrence of $\ev_2$ belongs to $\evtr_2$; formally, this
  correponds to the functions $\fun{f_1}{\{0,1\}}{\{0,3\}}$ and $\fun{f_2}{\{0,1\}}{\{1,2\}}$ s.t. $f_1(0)=0$, $f_1(1)=3, f_2(0)=1, f_2(1)=2$, which
  satisfy the first item of the definition, but not the second, because $\min\{ n_1\in\{0,1\} \mid f_2(0)=1<f_1(n_1) \}=1$ and
  $\evtr_1(1)=\ev_2=\evtr_2(0)$; the functions $f'_1$ and $f'_2$ s.t. $f'_1(0)=0$, $f'_1(1)=1, f'_2(0)=3, f'_2(1)=2$ satisfy both items, but $f'_2$ is \textbf{not}
strictly increasing.}

\paragraph{Generalized left-preferential shuffle:}
Given a set of event traces $\evtrSet$, the generalized left-preferential shuffle $\evtr_1\glpshuffleop{\evtrSet}\evtr_2$ of event trace $\evtr_1$ and $\evtr_2$ w.r.t. $\evtrSet$ is the set of traces $\evtrSet'\subseteq\evtr_1\lpshuffleop\evtr_2$ s.t. $\evtr\in\evtrSet'$ iff
$\dom(\evtr)$ can be partitioned into $N_1$ and $N_2$ in such a way that there exist two strictly increasing and bijective partial functions
$\fun{f_1}{\dom(\evtr_1)}{N_1}$ and $\fun{f_2}{\dom(\evtr_2)}{N_2}$ s.t.
\begin{itemize}
\item $\evtr_1(n_1)=\evtr(f_1(n_1))$ and $\evtr_2(n_2)=\evtr(f_2(n_2))$, for all $n_1\in\dom(\evtr_1)$, $n_2\in\dom(\evtr_2)$;  
\item for all $n_2\in\dom(\evtr_2)$, if $m=\min\{ n_1\in\dom(\evtr_1) \mid f_2(n_2)<f_1(n_1) \}$, then $\evtr'(m)\neq\evtr_2(n_2)$ for all
$\evtr'\in\evtrSet$ s.t. $m\in\dom(\evtr')$.  
\end{itemize}  
From the definitions of the shuffle operators above one can easily deduce that
$\evtr_1\glpshuffleop{\emptyset}\evtr_2=\evtr_1\shuffleop\evtr_2$ and $\evtr_1\glpshuffleop{\{\evtr_1\}}\evtr_2=\evtr_1\lpshuffleop\evtr_2$, 
for all event traces $\evtr_1$, $\evtr_2$.
\added{This generalisation of the left-preferential shuffle is needed to define the compositional semantics of the shuffle in Section~\ref{sec:comp}. Let us consider $T_1=\{\ev_1\catop\ev_2,\ev_3\catop\ev_4\}$ and $T_2=\{\ev_1\catop\ev_5\}$; one might be tempted to define
  $T_1 \lpshuffleop T_2$ as the set $\{ \evtr \;|\; \evtr_1\in T_1, \evtr_2\in T_2, \evtr\in\evtr_1\lpshuffleop\evtr_2\}$, which corresponds to $\{ \ev_1\catop\ev_2\catop\ev_1\catop\ev_5, \ev_1\catop\ev_1\catop\ev_2\catop\ev_5, \ev_1\catop\ev_1\catop\ev_5\catop\ev_2, \ev_3\catop\ev_4\catop\ev_1\catop\ev_5, \ev_3\catop\ev_1\catop\ev_4\catop\ev_5, \ev_3\catop\ev_1\catop\ev_5\catop\ev_4, \ev_1\catop\ev_5\catop\ev_3\catop\ev_4, \ev_1\catop\ev_3\catop\ev_4\catop\ev_5, \ev_1\catop\ev_3\catop\ev_5\catop\ev_4 \}$. But, the last three traces, where $\ev_1$ is consumed from $T_2$ as first event, are not correct,
  because the event $\ev_1$ in $T_1$ must take the precedence. Thus, the correct definition is given by 
$\{ \evtr \;|\; \evtr_1\in T_1, \evtr_2\in T_2, \evtr\in\evtr_1\glpshuffleop{\evtrSet_1}\evtr_2\}$, which does not contain the three traces mentioned above.} 

%% file: calculus.tex
\section{The \rml trace calculus}\label{sec:calculus}

In this section we define the operational semantics of the trace calculus
on which \rml is based on.
An \rml specification is compiled into a term of the trace calculus, which is used as an Intermediate Representation, and then
a SWI-Prolog\footnote{\url{http://www.swi-prolog.org/}}
monitor is generated; its execution employs the interpreter of the trace calculus, whose SWI-Prolog implementation
is directly driven by the reduction rules defining the labeled transition system of the calculus.

\paragraph{Syntax.} 
The syntax of the calculus is defined in Figure~\ref{fig:syn}.
\begin{figure}
\input{syntax}
\caption{Syntax of the \rml trace calculus: $\eventTy$ is defined inductively, $\te$ is defined coinductively on the set of cyclic terms.}\label{fig:syn}
\end{figure}
The main basic building block of the calculus is provided by the notion of \emph{event type pattern}, an expression 
consisting of a name $\evTyName$ of an \emph{event type}, applied to arguments which are \emph{basic data expressions}
denoting either variables or the data values (of primitive, array, or object type) associated with the events perceived by the monitor.
An event type is a predicate which defines a possibly infinite set of events; an event type pattern specifies the set of events that are expected to
occur at a certain point in the event trace; since event type patterns
can contain variables, upon a successful match a substitution is computed to bind the variables of the pattern with the data values carried by the
matched event.

\rml is based on a general object model where events are represented as 
JavaScript object literals;
for instance, the event type \lstinline{open($\mathit{fd}$)} of arity 1 may represent all events 
stating `function call \lstinline{fs.open} has returned file descriptor  $\mathit{fd}$' and having shape
\lstinline!{event:'func_post', name:'fs.open', res:$\mathit{fd}$}!.
The argument $\mathit{fd}$ consists of the file descriptor (an integer value)
returned by a call to \lstinline{fs.open}. The definition is parametric in the variable $\mathit{fd}$
which can be bound only when the corresponding event is matched with the information of the
file descriptor associated with the property \lstinline{res}; for instance, \lstinline{open(42)}
matches all events of shape \lstinline!{event:'func_post', name:'fs.open', res:42}!, that is, all
returns from call to \lstinline{fs.open} with value 42.  

Despite \rml offers to the users the possibility to define the event types that are used in the specification,
for simplicity the calculus is independent of the language used to define event types; correspondingly,
the definition of the rewriting system of the calculus
is parametric in the relation $\mtch$ assigning a semantics to event types (see below).

A specification is represented by a trace expression $\te$ built on top of the constant
$\emptyTrace$ (denoting the singleton set with the empty trace),
event type patterns $\eventTy$ (denoting the sets of all traces of length 1 with events matching
$\eventTy$), the binary operators (able to combine together sets of traces) of concatenation (juxtaposition), intersection ($\andop$),
union ($\orop$) and shuffle ($\shuffleop$), and a let-construct to define the scope of variables used in event type patterns.

Differently from event type patterns, which are inductively defined terms, 
trace expressions are assumed to be cyclic (a.k.a. regular or rational) \cite{Courcelle83,FrischEtAl08,AnconaC14,AnconaC16}
to provide an abstract support to recursion, since no explicit constructor is needed for it:
the depth of a tree corresponding to a trace expression is allowed to be infinite, but the number of its different subtrees must be finite.
This condition is proved to be equivalent \cite{Courcelle83} to requiring that a trace expression
can always be defined by a \emph{finite} set\footnote{The internal representation of cyclic terms in SWI-Prolog is indeed based on such approach.} of possibly recursive syntactic equations.

Since event type patterns are inductive terms, the definition of free variables for them is standard.
\begin{definition}\label{def:fv}
The set of free variables $\pfv(\eventTy)$ occurring in an event type pattern $\eventTy$ is inductively defined as follows:
$$
\begin{array}{l}
\pfv(\dvar)=\{\dvar\} \qquad \pfv(\pl)=\emptyset \\  
\pfv(\evTyName(\bv_1,\ldots,\bv_n))=\pfv(\{ \pk_1{:}\bv_1,\ldots,\pk_n{:}\bv_n \})=\pfv([\bv_1,\ldots,\bv_n])=\bigcup_{i=1\ldots n}\pfv(\bv_i)
\end{array}
$$
\end{definition}

Given their cyclic nature, a similar inductive definition of free variables for trace expressions does not work;
for instance, if \lstinline[basicstyle=\ttfamily\normalsize]{$\te=$open($\mathit{fd}$)$\catop\te$}, a definition of $\fv$ given by induction
on trace expressions would work only for non-cyclic terms and would be undefined for $\fv(\te)$. Unfortunately,
neither a coinductive definition could work correctly since the set $S$ returned by $\fv(\te)$ has to satisfy
the equation $S=\{\mathit{fd}\}\cup S$ which has \added{infinitely many solutions}; hence, while an inductive definition of $\fv$ leads to
a partial function which is undefined for all cyclic terms, a coinductive definition results in a non-functional relation $\fv$;
luckily, such a relation always admits the ``least solution'' which corresponds to the intended semantics.

\begin{fact}
Let $\fvp$ be the predicate on trace expressions and set of variables, coinductively defined as follows:
$$
\begin{array}{l}
  \RuleNoName{}{\fvp(\emptyTrace,\emptyset)}{} \qquad \RuleNoName{}{\fvp(\eventTy,S)}{\pfv(\eventTy)=S} \qquad
  \RuleNoName{\fvp(\te,S)}{\fvp(\block{\dvar}{\te},S\setminus\{\dvar\})}{}\qquad
  \RuleNoName{\fvp(\te_1,S_1)\quad\fvp(\te_2,S_2)}{\fvp(\te_1\op\te_2,S_1\cup S_2)}{\op\in\{\shuffleop,\catop,\andop,\orop\}} 
\end{array}
$$
Then, for any trace expression $\te$, if $L=\bigcap\{S \mid \fvp(\te,S) \mbox{ holds}\}$, then $\fvp(\te,L)$ holds.    
\end{fact}
\begin{proof}
By case anaysis on $\te$ and coinduction on the definition of $\fvp(\te,S)$. 
\end{proof}
\begin{definition}
The set of free variables $\fv(\te)$ occurring in a trace expression is defined by $\fv(\te)=\bigcap\{S \mid \fvp(\te,S) \mbox{ holds}\}$.
\end{definition}

\paragraph{Semantics.}
\begin{figure}
\input{semantics}
\caption{Transition system for the trace calculus.}
\label{fig:semantics}
\end{figure}
The semantics of the calculus depends on three judgments, inductively defined by the inference rules in Figure \ref{fig:semantics}.
Events $\ev$ range over a fixed universe of events $\eventSet$.
The judgment $\der\isEmpty(\te)$ is derivable iff $\te$ accepts the empty trace $\emptyevtr$ and is auxiliary to
the definition of the other two judgments $\te_1\extrans{\ev}\te_2;\subs$ and $\te\notrans{\ev}$;
the rules defining it are straightforward and are independent from the remaining judgments, hence
a stratified approach is followed and $\der\isEmpty(\te)$ and its negation $\notder\isEmpty(\te)$ are safely used
in the side conditions of the rules for $\te_1\extrans{\ev}\te_2;\subs$ and $\te\notrans{\ev}$ (see below).

The judgment $\te_1\extrans{\ev}\te_2;\subs$ defines the single reduction steps of the labeled
transition system on which the semantics of the calculus is based;
$\te_1\extrans{\ev}\te_2;\subs$ is derivable iff the event $\ev$ can be consumed, with the generated substitution $\subs$, by the expression
$\te_1$, which then reduces to $\te_2$. The judgment $\te\notrans{\ev}$
is derivable iff there are no reduction steps for event $\ev$ starting from expression $\te$ and is needed to
enforce a deterministic semantics and to guarantee that the rules are monotonic and, hence, the existence of the least fixed-point;
the definitions of the two judgments are mutually recursive. 

Substitutions are finite partial maps from variables to data values which
are produced by successful matches of event type patterns; the domain of $\subs$ and the empty substitution are denoted by $\dom(\subs)$ and
$\emptysub$, respectively, while $\restrict{\subs}{\dvar}$ and $\remove{\subs}{\dvar}$ denote
the substitutions obtained from $\subs$ by restricting its domain to $\{\dvar\}$ and removing $\dvar$ from its domain, respectively.
We simply write $\te_1\extrans{\ev}\te_2$ to mean $\te_1\extrans{\ev}\te_2;\emptysub$.
Application of a substitution $\subs$ to an event type patter $\eventTy$ is denoted by $\subs\eventTy$, and defined by induction on $\eventTy$:
$$
\begin{array}{l}
\subs\dvar=\subs(\dvar)\mbox{ if $\dvar\in\dom(\subs)$, } \subs\dvar=\dvar \mbox{ otherwise} \qquad \subs\pl=\pl\\
\subs\{ \pk_1{:}\bv_1,\ldots,\pk_n{:}\bv_n \}=\{ \pk_1{:}\subs\bv_1,\ldots,\pk_n{:}\subs\bv_n \} \qquad
\subs[\bv_1,\ldots,\bv_n]=[\subs\bv_1,\ldots,\subs\bv_n] \\
\subs\evTyName(\bv_1,\ldots,\bv_n)=\evTyName(\subs\bv_1,\ldots,\subs\bv_n)
\end{array}
$$

Application of a substitution $\subs$ to a trace expression $\te$ is denoted by $\subs\te$, and defined by coinduction 
on $\te$:
$$
\begin{array}{l}
  \subs\emptyTrace=\emptyTrace \qquad \subs\eventTy=\subs\evTyName(\bv_1,\ldots,\bv_n) \mbox{ if $\eventTy=\evTyName(\bv_1,\ldots,\bv_n)$}  \\
  \subs(\te_1\op\te_2)=\subs\te_1\op\subs\te_2 \mbox{ for $\op\in\{\catop,\andop,\orop,\shuffleop\}$} \qquad \subs\block{\dvar}{\te}=\block{\dvar}{\remove{\subs}{\dvar}\te}
\end{array}
$$
Since the calculus does not cover event type definitions,
the semantics of event types is parametric in the auxiliary partial function $\mtch$, used in the side condition of rules
(prefix) and (n-prefix): $\mtch(\ev,\eventTy)$ returns the substitution $\subs$ iff event $\ev$ matches event type $\subs\eventTy$
and fails (that is, is undefined) iff there is no substitution $\subs$ for which $\ev$ matches $\subs\eventTy$. The substitution is expected
to be the most general one and, hence, its domain to be included in the set of free variables in $\eventTy$ (see Def.~\ref{def:fv}).

As an example of how $\mtch$ could be derived from the definitions of event types in \rml, if
we consider again the event type \lstinline{open($\mathit{fd}$)} and $\ev$=\lstinline!{event:'func_post', name:'fs.open', res:42}!, then
\lstinline{$\mtch(\ev,$open(fd)$)=\{$fd$\mapsto 42\}$}, while \lstinline{$\mtch(\ev,$open(23)$)$} is undefined.

Except for intersection, which is intrinsically deterministic since both operands need to be reduced, the rules defining the semantics of the other binary
operators depend on the judgment  $\te\notrans{\ev}$ to force determinism; in particular, the judgment is used to ensure a left-to-right evaluation
strategy: reduction of the right operand is possible only if the left hand side cannot be reduced. 

The side condition of rule (and) uses the partial binary operator $\subsMerge$ to merge substitutions:
$\subs_1\subsMerge\subs_2$ returns the union of 
$\subs_1$ and $\subs_2$, if they coincide on the intersection of their domains, and is undefined otherwise. 

Rule (cat-r) uses the judgment $\isEmpty(\te_1)$ in its side condition: event $\ev$ consumed by $\te_2$ 
can also be consumed by $\te_1\catop\te_2$ only if $\ev$ is not consumed by $\te_1$ (premise $\te_1\notrans{\ev}$
forcing left-to-right deterministic reduction), and the empty trace is accepted by $\te_1$ (side condition $\der\isEmpty(\te_1)$).

Rule (par-t) can be applied when variable $\dvar$ is in the domain of the 
substitution $\subs$ generated by the reduction step from $\te$ to $\te'$:
the substitution $\restrict{\subs}{\dvar}$ restricted to $\dvar$  is applied to $\te'$, and $\dvar$ is removed from the domain of $\subs$, together with its corresponding declaration.
If $\dvar$ is not in the domain of $\subs$ (rule (par-f)), no substitution and no declaration removal is performed.


The rules defining $\te\notrans{\ev}$ are complementary to those for $\te\extrans{\ev}\te'$, and the definition of $\te\notrans{\ev}$
depends on the judgment $\te\extrans{\ev}\te'$ because of rule (n-and):
there are no reduction steps for event $\ev$ starting from expression $\te_1\andop\te_2$, even when
$\te_1\extrans{\ev}\te'_1;\subs_1$ and $\te_2\extrans{\ev}\te'_2;\subs_2$ are derivable, if the two
generated substitutions $\subs_1$ and $\subs_2$ cannot be successfully merged together; this happens when
there are two event type patterns that match event $\ev$ for two incompatible values of the same variable. 

\added{Let us consider an example of a cyclic term with the let-construct: $\te=\block{\fd}{open(\fd) \catop close(\fd) \catop \te}$. The
  trace expression declares a local variable $\fd$ (the file descriptor), and requires that two immediately subsequent open and close events share the same
  file descriptor. Since the recursive occurrence of $\te$ contains a nested let-construct, the subsequent open and
  close events can involve a different file descriptor, and this can happen an infinite number of times. In terms of derivation,
  starting from $\te$, if the event \lstinline{\{event:'func_post', name:'fs.open', res:42\}} is observed, which matches \lstinline{open(42)}, then the substitution $\{\fd \mapsto 42\}$ is computed. As a consequence, the residual term $close(42) \catop \te$ is obtained, by substituting
  $\fd$ with $42$ and removing the let-block. After that, the only valid event which can be observed is \lstinline{\{event:'func_pre',name:'close',args:[42]\}}, matching $close(\fd)$. Thus, after this rewriting step we get $\te$ again; the behavior continues as before, but
  a different file descriptor can be matched because of the let-block which hides the outermost declaration of $\fd$; indeed,
  the substitution is not propagated inside the nested let-block. Differently from $\te$, the term $\block{\fd}{\te'}$ with
  $\te'=open(\fd) \catop close(\fd) \catop \te'$ would require all open and close events to match a unique global file descriptor. As
  further explained in Section~\ref{sect:related_work}, such example shows how the let-construct is a solution more flexible than the mechanism of trace slicing
used in other RV tools to achieve parametricty.}

The following lemma can be proved by induction on the rules defining $\te\extrans{\ev}\te';\subs$.
\begin{lemma}\label{lemma:vars}
If $\te\extrans{\ev}\te';\subs$ is derivable, then $\dom(\subs)\cup\fv(\te')\subseteq\fv(\te)$.
\end{lemma}
Since trace expressions are cyclic, they can only contain a finite set of free variables,
therefore the domains of all substitutions generated by a possibly infinite sequence of consecutive reduction steps
starting from $\te$ are all contained in $\fv(\te)$.

\subsection{Semantics based on the transition system}

The reduction rules defined above provide the basis for the semantics of the calculus; because of computed substitutions and free variables, 
the semantics of a trace expression is not just a set of event traces: every accepted trace must be equipped with a substitution specifying
how variables have been instantiated during the reduction steps. We call it an \emph{instantiated event trace};
this can be obtained from the pairs of event and substitution traces yield by the possibly infinite reduction steps, by considering the
disjoint union of all returned substitutions. \added{Such a notion is needed\footnote{See the example in Section~\ref{sec:comp}.} to allow a compositional semantics.} The notion of substitution trace can be given in an analogous way as done for event traces in
Section~\ref{sec:back}.
By the considerations related to Lemma~\ref{lemma:vars}, the substitution associated with an instantiated event trace has always a finite domain,
even when the trace is infinite; this means that the
substitution is always fully defined after a finite number of reduction steps.

\begin{definition}
A  \emph{concrete instantiated event trace} over the event universe $\eventSet$ is a pair $(\evtr,\overline{\subs})$ of event traces over $\eventSet$,
and substitution traces s.t. either  $\evtr$ and $\overline{\subs}$ are both infinite, or they are both finite and have the same length, all
the substitutions in  $\overline{\subs}$ have mutually disjoint domains and $\bigcup\{\dom(\subs') \mid \subs'\in\substr\}$ is finite.

An \emph{abstract instantiated event trace} (instantiated event trace, for short) over $\eventSet$ is a pair $(\evtr,\subs)$
where $\evtr$ is an event trace over $\eventSet$ and $\subs$ is a substitution.
We say that $(\evtr,\subs)$ is derived from the concrete instantiated event trace $(\evtr,\substr)$, written $(\evtr,\substr)\leadsto(\evtr,\subs)$, iff
$\subs=\bigcup\{\subs' \mid \subs'\in\substr\}$.
\end{definition}
In the rest of the paper we use the meta-variable $\its$ to denote sets of instantiated event traces.
We use the notations $\tproj{\its}$ and $\sproj{\its}$ to denote the two projections
$\{\evtr \mid (\evtr,\subs)\in\its \}$ and $\{\subs \mid (\evtr,\subs)\in\its \}$, respectively;
we write $\evtr\pref\its$ to mean $\evtr\pref\tproj{\its}$.
The notation $\inft{\its}$ denotes the set $\{(\evtr,\subs) \mid (\evtr,\subs) \in \its, \evtr\ \mathit{infinite} \}$ restricted to infinite traces.


We can now define the semantics of trace expressions.
\begin{definition}\label{def:sem}
The concrete semantics $\csem{\te}$ of a trace expression $\te$ is the set of concrete instantiated event traces coinductively defined as follows:
\begin{itemize}
\item $(\emptyevtr,\emptysubtr)\in\csem{\te}$ iff $\der\isEmpty(\te)$ is derivable;
\item $(\ev\catop\evtr,\subs\catop\substr)\in\csem{\te}$ iff $\te\extrans{\ev}\te';\subs$ is derivable and $(\evtr,\substr)\in\csem{\subs\te'}$.
\end{itemize}
The (abstract) semantics $\sem{\te}$ of a trace expression $\te$ is the set of instantiated event traces
$\{(\evtr,\subs) \mid  (\evtr,\substr)\in\csem{\te}, (\evtr,\substr)\leadsto(\evtr,\subs) \}$.

\end{definition}

The following propositions show that the concrete semantics of a trace expression $\te$ as given in Definition~\ref{def:sem} is always well-defined.

\begin{proposition}\label{prop:sem-one}
If $(\evtr,\substr)\in\csem{\te}$ and $\evtr$ is finite, then $\lgth{\evtr}=\lgth{\substr}$.
\end{proposition}

\begin{proposition}\label{prop:sem-two}
If $(\evtr,\substr)\in\csem{\te}$ and $\evtr$ is infinite, then $\substr$ is infinite as well.
\end{proposition}

\begin{proposition}\label{prop:sem-three}
If $(\evtr,\substr)\in\csem{\te}$, then for all $n,m\in\nat$, $n\neq m$ implies $\dom(\substr(n))\cap\dom(\substr(m))=\emptyset$.
\end{proposition}

\begin{proposition}\label{prop:sem-four}
If $(\evtr,\substr)\in\csem{\te}$, then for all $n\in\nat$ $\dom(\substr(n))\subseteq\fv(\te)$.
\end{proposition}

%% file: syntax.tex
 \begin{center} 
\begin{small}
\[
\begin{array}{rcll}
\production{\dv}{\pl \mid \{ \pk_1{:}\dv_1,\ldots,\pk_n{:}\dv_n \} \mid [\dv_1,\ldots,\dv_n]}{(data value)}
\production{\bv}{\dvar \mid \pl \mid \{ \pk_1{:}\bv_1,\ldots,\pk_n{:}\bv_n \} \mid [\bv_1,\ldots,\bv_n]}{(basic data expression)}  
\production{\eventTy}{\evTyName(\bv_1,\ldots,\bv_n)}{(event type pattern)}
\production{\te}{\emptyTrace}{(empty trace)}
\productionNl{\eventTy}{(single event)}
\productionNl{\mid \te_1\catop\te_2}{(concatenation)}
\productionNl{\mid \te_1\andop\te_2}{(intersection)}
\productionNl{\mid \te_1\orop\te_2}{(union)}
\productionNl{\mid \te_1\shuffleop\te_2}{(shuffle)}
\productionNl{\mid \block{\dvar}{\te}}{(parametric expression)}
\end{array}
\]
\end{small}
 \end{center}

%% file: semantics.tex
\begin{center}
\begin{small}  
\begin{math}
\begin{array}{c}
\Rule{e-$\epsilon$}
{}
{\der\isEmpty(\emptyTrace)}
{}
\quad
\Rule{e-al}
{\der\isEmpty(\te_1)\quad\der\isEmpty(\te_2)}
{\der\isEmpty(\te_1 \op \te_2)}
{\op\in\{\shuffleop,\catop,\andop\}}
\quad
\Rule{e-or-l}
{\der\isEmpty(\te_1)}
{\der\isEmpty(\te_1\orop\te_2)}
{}
\quad
\Rule{e-or-r}
{\der\isEmpty(\te_2)}
{\der\isEmpty(\te_1\orop\te_2)}
{}
\\[2ex]
\Rule{e-par}
{\der\isEmpty(\te)}
{\der\isEmpty(\block{\dvar}{\te})}
{}
\quad
\Rule{single}
{}
{\eventTy\extrans{\ev}\emptyTrace;\subs}
{
  \subs=\mtch(\ev,\eventTy)
}
\quad
\Rule{or-l}
{\te_1\extrans{\ev}\te'_1;\subs}
{\te_1\orop\te_2\extrans{\ev}\te'_1;\subs}
{}
\quad
\Rule{or-r}
{\te_1\notrans{\ev}\quad\te_2\extrans{\ev}\te'_2;\subs}
{\te_1\orop\te_2\extrans{\ev}\te'_2;\subs}
{}
\\[2ex]
\Rule{and}
{\te_1\extrans{\ev}\te'_1;\subs_1\quad\te_2\extrans{\ev}\te'_2;\subs_2}
{\te_1\andop\te_2\extrans{\ev}\te'_1\andop\te'_2;\subs}
{\subs=\subs_1\subsMerge\subs_2}
\quad
\Rule{shuffle-l}
{\te_1\extrans{\ev}\te'_1;\subs}
{\te_1\shuffleop\te_2\extrans{\ev}\te'_1\shuffleop\te_2;\subs}
{}
\quad
\Rule{shuffle-r}
{\te_1\notrans{\ev}\quad\te_2\extrans{\ev}\te'_2;\subs}
{\te_1\shuffleop\te_2\extrans{\ev}\te_1\shuffleop\te'_2;\subs}
{}
\\[2ex]
\Rule{cat-l}
{\te_1\extrans{\ev}\te'_1;\subs}
{\te_1\catop\te_2\extrans{\ev}\te'_1\catop\te_2;\subs}
{}
\quad
\Rule{cat-r}
{\te_1\notrans{\ev}\quad\te_2\extrans{\ev}\te'_2;\subs}
{\te_1\catop\te_2\extrans{\ev}\te'_2;\subs}
{\der\isEmpty(\te_1)}
\quad
\Rule{par-t}
{\te\extrans{\ev}\te';\subs}
{\block{\dvar}{\te}\extrans{\ev}\restrict{\subs}{\dvar}\te';\remove{\subs}{\dvar}}
{\dvar\in\dom(\subs)}
\\[2ex]
\Rule{par-f}
{\te\extrans{\ev}\te';\subs}
{\block{\dvar}{\te}\extrans{\ev}\block{\dvar}{\te'};\subs}
{\dvar\not\in\dom(\subs)}
\quad
\Rule{n-$\epsilon$}
{}
{\emptyTrace\notrans{\ev}}
{}
\quad
\Rule{n-single}
{}
{\eventTy\notrans{\ev}}
{
  \mtch(\ev,\eventTy)\ \text{undef} 
}
\quad
\Rule{n-or}
{\te_1\notrans{\ev}\quad\te_2\notrans{\ev}}
{\te_1\orop\te_2\notrans{\ev}}
{}
\\[2ex]
\Rule{n-and-l}
{\te_1\notrans{\ev}}
{\te_1\andop\te_2\notrans{\ev}}
{}
\quad
\Rule{n-and-r}
{\te_2\notrans{\ev}}
{\te_1\andop\te_2\notrans{\ev}}
{}
\quad
\Rule{n-and}
{\te_1\extrans{\ev}\te'_1;\subs_1\quad\te_2\extrans{\ev}\te'_2;\subs_2}
{\te_1\andop\te_2\notrans{\ev}}
{\subs_1\subsMerge\subs_2\ \text{undef}}
\\[2ex]
\Rule{n-shuffle}
{\te_1\notrans{\ev}\quad\te_2\notrans{\ev}}
{\te_1\shuffleop\te_2\notrans{\ev}}
{}
\quad
\Rule{n-cat-l}
{\te_1\notrans{\ev}}
{\te_1\catop\te_2\notrans{\ev}}
{\notder\isEmpty(\te_1)}
\quad
\Rule{n-cat-r}
{\te_1\notrans{\ev}\quad\te_2\notrans{\ev}}
{\te_1\catop\te_2\notrans{\ev}}
{}
\quad
\Rule{n-par}
{\te\notrans{\ev}}
{\block{\dvar}{\te}\notrans{\ev}}
{}
\end{array}
\end{math}
\end{small}
\end{center}

%% file: comp.tex
\section{Towards a compositional semantics}\label{sec:comp}

In this section we show how each basic trace expression operator can be interpreted as
an operation over sets of instantiated event traces and we formally prove that such an interpretation is
equivalent to the semantics derived from the transition system of the calculus in Definition~\ref{def:sem}, if one considers only
contractive terms.

\subsection{Composition operators}\label{sec:op-sem}

\paragraph{Left-preferential union:}
The left-preferential union $\its_1\corop\its_2$ of sets of instantiated event traces $\its_1$ and $\its_2$ is defined as follows:
$
\its_1\corop\its_2 = \its_1 \bigcup \{(\evtr,\subs) \in \its_2 \mid \evtr=\emptyevtr \mbox{ or } (\evtr=\ev\catop\evtr',\ev \nopref \its_1) \}.
$

In the deterministic left-preferential version of union, instantiated event traces in $\its_2$ are kept only if they start
with an event which is not the first element of \added{any of the traces} in $\its_1$ (the condition vacuously holds for the empty trace); since reduction steps can involve only one of the two operands at time, no restriction on the substitutions of the instantiated event traces is required.

\paragraph{Left-preferential concatenation:}
The left-preferential concatenation $\its_1\ccatop\its_2$ of sets of instantiated event traces $\its_1$ and $\its_2$ is defined as follows:
$
\its_1\ccatop\its_2 = 
\inft{\its_1} \cup 
\{ (\evtr_1\catop\evtr_2,\subs) \mid (\evtr_1,\subs_1) \in \its_1, (\evtr_2,\subs_2) \in \its_2, \subs=\subs_1\subsMerge\subs_2, (\evtr_2=\emptyevtr \mbox{ or } (\evtr_2=\ev\catop \evtr_3, (\evtr_1\catop\ev) \nopref \its_1)) \}.
$

The left operand $\inft{\its_1}$ of the union corresponds to the fact that 
in the deterministic left-preferential version of concatenation, all infinite instantiated event traces
in $\its_1$ belong to the semantics of concatenation. The right operand of the union specifies the behavior
for all finite instantiated event traces $\evtr_1$ in $\its_1$; in such cases, the trace in $\its_1\ccatop\its_2$ can continue with 
$\evtr_2$ in $\its_2$ if $\evtr_1$ is not allowed to continue in $\its_1$ with
the first event $\ev$ of $\evtr_2$ ($(\evtr_1\catop\ev) \nopref \its_1$, the condition vacuously holds if $\evtr_2$ is the empty trace).
Since the reduction steps corresponding to $\evtr_2$ follow those for $\evtr_1$, the overall substitution $\subs$ must
meet the constraint $\subs=\subs_1\subsMerge\subs_2$ ensuring that $\subs_1$ and $\subs_2$ match on the shared variables 
of the two operands.
\paragraph{Intersection:}
The intersection $\its_1\candop\its_2$ of sets of instantiated event traces $\its_1$ and $\its_2$ is defined as follows:
$
\its_1\candop\its_2 = 
\{(\evtr,\subs) \mid (\evtr,\subs_1)\in\its_1, (\evtr,\subs_2)\in\its_2,  \subs=\subs_1\subsMerge\subs_2\}.
$

Since intersection is intrinsically deterministic, its semantics throws no surprise. 

\paragraph{Left-preferential shuffle:}
The left-preferential shuffle $\its_1\cshuffleop\its_2$ of sets of instantiated event traces $\its_1$ and $\its_2$
is defined as follows:
$
\its_1\cshuffleop\its_2 = \{(\evtr,\subs) \mid (\evtr_1,\subs_1) \in \its_1, (\evtr_2,\subs_2) \in \its_2, \subs=\subs_1\subsMerge\subs_2,
\evtr\in\evtr_1\glpshuffleop{\tproj{\its_1}}\evtr_2 \}.
$

The definition is based on the generalized left-preferential shuffle defined in Section~\ref{sec:back};
an event in $\evtr_2$ at a certain position $n$ can contribute to the shuffle only if no trace in
$\tproj{\its_1}$ could contribute with the same event at the same position $n$. 
Since the reduction steps corresponding to $\evtr_1$ and $\evtr_2$ are interleaved, the overall substitution $\subs$ must
meet the constraint $\subs=\subs_1\subsMerge\subs_2$ ensuring that $\subs_1$ and $\subs_2$ match on the shared variables 
of the two operands.

\paragraph{Variable deletion:}

The deletion $\remove{\its}{\dvar}$ of $\dvar$ from the set of instantiated event traces $\its$
is defined as follows:
$
\remove{\its}{\dvar} = \{ (\evtr,\remove{\subs}{\dvar}) \mid (\evtr,\subs) \in \its \}.
$

As expected, variable deletion only affects the domain of the computed substitution. 

\added{
  The definitions above show that instantiated event traces are needed to allow a compositional semantics; let us consider
  the following simplified variation of the example given in Section~\ref{sec:calculus}:
  $\te'=\block{\fd}{open(\fd)\catop close(\fd)}$. If we did not keep track of substitutions, then the compositional semantics of
  $open(\fd)$ and $close(\fd)$ would contain all traces of length 1 matching $open(\fd)$ and $close(\fd)$, respectively, for any value
  $\fd$, and, hence,  the semantics of $open(\fd)\catop close(\fd)$ could not constrain $open$ and $close$ events to be on the
  same file descriptor. Indeed, such a constraint is obtained by checking that the substitution of the event trace matching
  $open(\fd)$ can be successfully merged with the substitution of the event trace matching $close(\fd)$, so that the two substitutions agree
on $\fd$.}

\subsection{Contractivity}
Contractivity is a condition on trace expressions which is statically enforced by the \rml compiler;
such a requirement avoids infinite loops when an event does not match
the specification and the generated monitor would try to build an infinite derivation.
Although the generated monitors could dynamically check potential loops dynamically, a syntactic condition
enforced statically by the compiler allows monitors to be relieved of such a check, and, thus, to be more efficient. 

Contractivity can be seen as a generalization of absence of left recursion in grammars \cite{left-rec}; loops
in cyclic terms are allowed only if they are all guarded by a concatenation where the left operand $\te$
cannot contain the empty trace (that is, $\notder\isEmpty(\te)$ holds), and the loop continues in
the right operand of the concatenation. If such a condition holds, then it is not possible to build
infinite derivations for  $\te_1\extrans{\ev}\te_2$.

Interestingly enough, such a condition is also needed to prove that the interpretation of operators
as given in Section~\ref{sec:op-sem} is equivalent to the semantics given in Definition~\ref{def:sem}.
Indeed, the equivalence result proved in Theorem~\ref{theo} is based on Lemma~\ref{lemma1}
stating that for all contractive term $\te_1$ and event $\ev$, there exist $\te_2$ and $\subs$ s.t.  $\te_1\extrans{\ev}\te_2;\subs$ is derivable
if and only if $\te_1\notrans{\ev}$ is not derivable; such a claim does not hold for a non contractive term as $\te=\te\orop\te$, because for all
$\ev$, $\te'$ and $\subs$, $\te\extrans{\ev}\te';\subs$ and $\te\notrans{\ev}$ are not derivable. This is due to the fact that both judgments are defined by an inductive inference system.
\added{Intuitively, from a contractive term we cannot derive a new term without passing through at least one concatenation. For instance, considering the term $\te=\ev\catop\te$, we have contractivity because we have to consume $\ev$ before going inside the loop. But, if we swap the operands, we obtain instead $\te=\te\catop\ev$, where contractivity does not hold; in fact, deriving the concatenation we go first inside the head, but it is cyclic. Since the $\extrans{}$ and $\notrans{}$ judgements are defined inductively, both are not derivable because a finite derivation tree cannot be derived for neither of them.}

\begin{definition}
  Syntactic contexts $\ctx$ are inductively defined as follows:
$$
\begin{array}{rcll}
\production{\ctx}{\Box\mid\ctx\op\te\mid\te\op\ctx\mid\block{\dvar}{\ctx}}{with $\op\in\{\andop,\orop,\shuffleop,\catop\}$}
\end{array}
$$
\end{definition}

\begin{definition}
  A  syntactic context $\ctx$ is contractive if one of the following conditions hold:
  \begin{itemize}
  \item $\ctx=\block{\dvar}{\ctx'}$ and $\ctx'$ is contractive;    
  \item $\ctx=\ctx'\op\te$, $\ctx'$ is contractive and $\op\in\{\catop,\andop,\orop,\shuffleop\}$;    
  \item $\ctx=\te\op\ctx'$, $\ctx'$ is contractive and $\op\in\{\andop,\orop,\shuffleop\}$;    
  \item $\ctx=\te\catop\ctx'$, $\der\isEmpty(\te)$ and $\ctx'$ is contractive;    
  \item $\ctx=\te\catop\ctx'$ and $\notder\isEmpty(\te)$.
  \end{itemize}
\end{definition}

\begin{definition}
A term is part of $\te$ iff it belongs to the least set  $\partof(\te)$ matching the following definition:
$$
\begin{array}{l}
  \partof(\emptyTrace)=\partof(\eventTy)=\emptyset \quad \partof(\block{\dvar}{\te})=
  \{\te\} \cup \partof(\te)\\
  \partof(\te_1\op\te_2)=\{\te_1,\te_2\}\cup\partof(\te_1)\cup\partof(\te_2)  \mbox{ for $\op\in\{\shuffleop,\catop,\andop,\orop\}$}
\end{array}
$$
\end{definition}
Because trace expressions can be cyclic, the definition of $\partof$ follows the same pattern adopted for $\fv$.
One can prove that a term $\te$ is cyclic iff there exists $\te'\in\partof(\te)$ s.t. $\te'\in\partof(\te')$. 

\begin{definition}
  A term $\te$ is \emph{contractive} iff the following conditions old:
  \begin{itemize}
  \item  for any syntactic context $\ctx$, if $\te=\ctx[\te]$ then $\ctx$ is contractive;
  \item  for any term $\te'$, if $\te'\in\partof(\te)$, then $\te'$ is contractive.
  \end{itemize}
\end{definition}

\subsection{Main Theorem}

We first list all the auxiliary lemmas used to prove Theorem~\ref{theo}.

\begin{lemma}\label{lemma1}
For all contractive term $\te_1$ and event $\ev$, there exist $\te_2$ and $\subs$ s.t.  $\te_1\extrans{\ev}\te_2;\subs$ is derivable
if and only if $\te_1\notrans{\ev}$ is not derivable.
\end{lemma}

\begin{lemma}\label{lemmaLeadsToVar}
 If $(\evtr,\substr)\leadsto(\evtr,\subs)$, then $(\evtr,\remove{\substr}{\dvar})\leadsto(\evtr,\remove{\subs}{\dvar})$.
\end{lemma}
Where $\remove{\substr}{\dvar}$ denotes the substitution sequence where $\dvar$ is removed from the domain of each substitution in $\substr$.

\begin{lemma}\label{lemmaSubsDecomposition}
 Given a substitution function $\subs$ and a term $\te$, we have that $\subs \te = \remove{\subs}{\dvar}\restrict{\subs}{\dvar}\te = \restrict{\subs}{\dvar}\remove{\subs}{\dvar}\te$, for every $\dvar\in\dom(\subs)$.
\end{lemma}

\begin{lemma}\label{lemma3}
Let $\te$ be a term, $\subs_1$ be a substitution function s.t. $\dom(\subs_1)=\{\dvar\}$; we have that: 
$$\forall_{(\evtr,\subs_2)\in\sem{\te}}.((\subs_1\subsMerge\subs_2 \text{ is defined}) \implies (\evtr,\remove{\subs_2}{\dvar})\in\sem{\subs_1\te})$$
\end{lemma}

\begin{lemma}\label{lemma4}
Let $\te$ be a term, $\subs_1$ be a substitution function s.t. $\dom(\subs_1)=\{\dvar\}$; we have that: 
$$\forall_{(\evtr,\subs_2)\in\sem{\subs_1\te}}.((\subs_1\subsMerge\subs_2 \text{ is defined}) \implies (\evtr,\subs_2)\in\sem{\te})$$
\end{lemma}

\begin{lemma}\label{lemma5}
$\te\notrans{\ev}\iff\ev\nopref\sem{\te}$.
\end{lemma}

\begin{lemma}\label{lemmaSubsRemoval}
 If $(\evtr,\subs)\in\sem{\te}$, then $(\evtr,\emptyset)\in\sem{\subs\te}$.
\end{lemma}

\begin{lemma}\label{lemmaOmega}
If $(\evtr,\substr)\in\csem{\te}$ and $\evtr$ is infinite, then $(\evtr,\substr)\in\csem{\te\catop\te'}$ for every $\te'$.
\end{lemma}


\begin{lemma}\label{lemmaHeadShuffle}
 If $\ev\catop\evtr \in \evtr_1\glpshuffleop{T}\evtr_2$, then $\evtr_1=\ev\catop\evtr_1'$, or $\evtr_2=\ev\catop\evtr_2'$ and $\ev\nopref\evtr_1$.
\end{lemma}

\begin{lemma}\label{lemmaEmptyTail}
If $(\evtr,\substr)\in\csem{\te}$ and $\isEmpty(\te')$, then $(\evtr,\substr)\in\csem{\te\catop\te'}$.
\end{lemma}

\begin{lemma}\label{lemmaConcatExpansion}
Given $(\evtr_1,\substr_1)\in\csem{\te_1}$, $\te_2\trans{\ev_2}\te_2^1;\subs_2^1$ and $(\evtr_2,\substr_2^2)\in\csem{\subs_2^1\te_2^1}$ with $\substr_2=\subs_2^1\catop\substr_2^2$. If $\evtr_1=\ev_1\catop\ldots\catop\ev_n$ is finite, $\te_1\trans{\ev_1}\te_1^1;\subs_1^1$, $\te_1^1\trans{\ev_2}\te_1^2;\subs_1^2$, $\ldots$, $\te_1^{n-1}\trans{\ev_n}\te_1^n;\subs_1^n$, with $\subs_1=\subs_1^1\catop\ldots\catop\subs_1^n$ and
$\te_1^n\notrans{\ev_2}$, then $(\evtr_1\catop\ev_2\catop\evtr_2, \substr_1\catop\substr_2)\in\csem{\te_1\catop\te_2}$.
\end{lemma}

In Theorem~\ref{theo} we claim that for every operator of the trace calculus, the compositional semantics is equivalent to the abstract semantics. To prove such claim, we need to show that, for each operator, every trace belonging to the compositional semantics belongs to the abstract semantics, which means we only consider correct traces (\emph{soundness}); and, every trace belonging to the abstract semantics belongs to the compositional semantics, which means we consider all the correct traces (\emph{completeness}).

Each operator requires a customised proofs, but in principle, all the proofs follow the same reasoning. Both soundness and completeness proof start expanding the compositional semantics definition in terms of its concrete semantics, which in turn is rewritten in terms of the operational semantics. At this point, the compositional operator's operands can be separately analysed in order to be recombined with the corresponding trace calculus operator. Finally, the proofs are concluded going backwards from the operational semantics to the abstract one, through the concrete semantics. For all the operators, except $\orop$ and $\andop$, the proofs are given by coinduction over the terms structure.
In every proof which is not analysed separately ($\iff$ cases), we implicitly apply Lemma~\ref{lemma1}.

\begin{theorem}\label{theo}
  The following claims hold for all contractive terms $\te_1$ and $\te_2$:
  \begin{itemize}
    \item $\sem{\te_1\orop\te_2}=\sem{\te_1}\corop\sem{\te_2}$
    \item $\sem{\te_1\catop\te_2}=\sem{\te_1}\ccatop\sem{\te_2}$
    \item $\sem{\te_1\andop\te_2}=\sem{\te_1}\candop\sem{\te_2}$
    \item $\sem{\te_1\shuffleop\te_2}=\sem{\te_1}\cshuffleop\sem{\te_2}$
    \item $\sem{\block{\dvar}{\te_1}}=\remove{\sem{\te_1}}{\dvar}$     
  \end{itemize}
\end{theorem}

The proofs for the union, intersection, shuffle and let cases are omitted and can be found in the extended version \cite{ancona2020determinism}. We decided not to report them due to space constraints. In the proofs that follow, we prove composed implications such as $A_1 \vee \ldots \vee A_n \implies B$, by splitting them into $n$ separate implications $A_1 \implies^1\; B$, $\ldots$, $A_n \implies^n\; B$.

The first operator we are going to analyse is the concatenation, where we are going to show that $(\evtr,\subs)\in\sem{\te_1\catop\te_2}\iff(\evtr,\subs)\in\sem{\te_1}\ccatop\sem{\te_2}$. 

The proof for the empty trace is trivial, and is constructed on top of the definition of the $\isEmpty$ predicate.

\begin{eqnarray*}
 (\emptyevtr,\emptyset)\in\sem{\te_1\catop\te_2} & \iff & (\emptyevtr,\emptyevtr)\in\csem{\te_1\catop\te_2} \wedge (\emptyevtr,\emptyevtr)\leadsto(\emptyevtr,\emptyset) \;\; \text{(by definition of $\sem{\te}$)} \\
 & \iff & \isEmpty(\te_1\catop\te_2) \text{ is derivable} \;\; \text{(by definition of $\csem{\te}$)} \\
 & \iff & \isEmpty(\te_1) \text{ is derivable } \wedge \isEmpty(\te_2) \text{ is derivable} \;\; \text{(by definition of $\isEmpty(\te)$)} \\
 & \iff & (\emptyevtr,\emptyevtr)\in\csem{\te_1} \wedge (\emptyevtr,\emptyevtr)\in\csem{\te_2} \wedge (\emptyevtr,\emptyevtr)\leadsto(\emptyevtr,\emptyset) \;\; \text{(by definition of $\csem{\te}$)} \\
 & \iff & (\emptyevtr,\emptyset)\in\sem{\te_1} \wedge (\emptyevtr,\emptyset)\in\sem{\te_2} \;\; \text{(by definition of $\sem{\te}$)} \\
 & \iff & (\emptyevtr,\emptyset)\in(\sem{\te_1}\ccatop\sem{\te_2}) \;\; \text{(by definition of $\ccatop$)}
\end{eqnarray*}

When the trace is not empty, we present the procedure to prove \emph{completeness} ($\implies$) and \emph{soundness} ($\impliedby$), separately.

Let us start with \emph{completeness}. To prove it, we have to show that the abstract semantics $\sem{\te_1\catop\te_2}$ (based on the original operational semantics) is included in the composition of the abstract semantics $\sem{\te_1}$ and $\sem{\te_2}$, using $\ccatop$ operator. More specifically, in the first part of the proof ($\implies^1$), the first event of the trace belongs to the head of the concatenation. Thus, the head is expanded through operational semantics, causing the term to be rewritten into a concatenation, where the head is substituted with a new term. Since the concrete semantics has been defined coinductively, we can conclude that the proof is satisfied by the so derived concatenation by coinduction. Finally, the proof is concluded recombining the new concatenation in terms of $\ccatop$. The second part of the proof ($\implies^2$) does not require coinduction, since the trace belongs to the tail of the concatenation. Through the operational semantics, the concatenation is rewritten into the new tail, and the proof is straightforwardly concluded following the abstract semantics.

\begin{eqnarray*}
(\ev\catop\evtr,\subs)\in\sem{\te_1\catop\te_2} & \implies & (\ev\catop\evtr,\substr)\in\csem{\te_1\catop\te_2} \wedge (\ev\catop\evtr,\substr)\leadsto(\ev\catop\evtr,\subs) \;\; \text{(by definition of $\sem{\te}$)} \\
& \implies & \te_1\catop\te_2\extrans{\ev}\te';\subs' \text{ is derivable } \wedge (\evtr,\substr')\in\csem{\subs'\te'} \;\; \text{(by definition of $\csem{\te}$)} \\
& \implies & (\te_1\extrans{\ev}\te_1';\subs' \text{ is derivable } \wedge \te_1\catop\te_2\extrans{\ev}\te_1'\catop\te_2;\subs' \text{ is derivable } \wedge \\
& & (\evtr,\substr')\in\csem{\subs'(\te_1'\catop\te_2)}) \vee \\
& & (\te_1\notrans{\ev} \wedge \isEmpty(\te_1) \wedge \te_2\extrans{\ev}\te_2';\subs' \text{ is derivable } \wedge \te_1\catop\te_2\extrans{\ev}\te_2';\subs' \text{ is derivable} \wedge \\
& & (\evtr,\substr')\in\csem{\subs'\te_2'}) \;\; \text{(by operational semantics)} \\
& \implies^1 & \te_1\extrans{\ev}\te_1';\subs' \text{ is derivable } \wedge \te_1\catop\te_2\extrans{\ev}\te_1'\catop\te_2;\subs' \text{ is derivable } \wedge \\
& & (\evtr,\subs'')\in\sem{\subs'(\te_1'\catop\te_2)} \wedge (\evtr,\substr')\leadsto(\evtr,\subs'') \wedge \subs=\subs''\subsMerge\subs' \\ 
& & \text{(by definition of $\sem{\te}$)} \\
& \implies^1 & \te_1\extrans{\ev}\te_1';\subs' \text{ is derivable } \wedge \te_1\catop\te_2\extrans{\ev}\te_1'\catop\te_2 \text{ is derivable } \wedge \\
& & (\evtr,\subs'')\in\sem{\subs'\te_1'}\ccatop\sem{\subs'\te_2} \wedge (\evtr,\substr')\leadsto(\evtr,\subs'') \wedge \subs=\subs''\subsMerge\subs' \\ 
& & \text{(by coinduction over $\sem{\te}$)} \\
& \implies^1 & \te_1\extrans{\ev}\te_1';\subs' \text{ is derivable } \wedge \te_1\catop\te_2\extrans{\ev}\te_1'\catop\te_2 \text{ is derivable } \wedge \\
& & (\evtr_1,\subs_1'')\in\sem{\subs'\te_1'} \wedge (\evtr_2,\subs_2'')\in\sem{\subs'\te_2} \wedge \evtr=\evtr_1\catop\evtr_2 \wedge \\
& & (\evtr_2=\emptyevtr \vee (\evtr_2=\ev'\catop\evtr_3 \wedge \evtr_1\catop\ev'\nopref\sem{\subs'\te_1'}))\;\; \text{(by definition of $\ccatop$)} \\
& \implies^1 & \te_1\extrans{\ev}\te_1';\subs' \text{ is derivable } \wedge (\evtr_1,\substr_1)\in\csem{\subs'\te_1'} \wedge (\evtr_1,\substr_1)\leadsto(\evtr_1,\subs_1'') \wedge \\
& & (\evtr_2,\subs_2'')\in\sem{\subs'\te_2} \wedge \evtr=\evtr_1\catop\evtr_2 \wedge \\
& & (\evtr_2=\emptyevtr \vee (\evtr_2=\ev'\catop\evtr_3 \wedge \evtr_1\catop\ev'\nopref\sem{\subs'\te_1'}))\;\; \text{(by definition of $\sem{\te}$)} \\
& \implies^1 & (\ev\catop\evtr_1,\subs'\catop\substr_1)\in\csem{\te_1} \wedge (\evtr_1,\substr_1)\leadsto(\evtr_1,\subs_1'') \wedge \\
& & (\evtr_2,\subs_2'')\in\sem{\subs'\te_2} \wedge \evtr=\evtr_1\catop\evtr_2 \wedge \\
& & (\evtr_2=\emptyevtr \vee (\evtr_2=\ev'\catop\evtr_3 \wedge \evtr_1\catop\ev'\nopref\sem{\subs'\te_1'}))\;\; \text{(by definition of $\csem{\te}$)} \\
& \implies^1 & (\ev\catop\evtr_1,\subs_1'')\in\sem{\subs'\te_1} \wedge (\evtr_2,\subs_2''\subsMerge\subs')\in\sem{\te_2} \wedge \evtr=\evtr_1\catop\evtr_2 \wedge \\
& & (\evtr_2=\emptyevtr \vee (\evtr_2=\ev'\catop\evtr_3 \wedge \evtr_1\catop\ev'\nopref\sem{\subs'\te_1'}))\\
& &  \text{(by definition of $\sem{\te}$ and Lemma~\ref{lemmaSubsRemoval})} \\
& \implies^1 & (\ev\catop\evtr,\subs)\in\sem{\te_1}\ccatop\sem{\te_2} \;\; \text{(by definition of $\ccatop$)} \\
& \implies^2 & (\ev\catop\evtr,\substr)\in\csem{\te_2} \wedge (\emptyevtr,\emptyevtr)\in\csem{\te_1} \wedge \te_1\notrans{\ev} \;\; \text{(by definition of $\csem{\te}$)} \\
& \implies^2 & (\ev\catop\evtr,\subs)\in\sem{\te_2} \wedge (\emptyevtr,\emptyset)\in\sem{\te_1} \wedge \te_1\notrans{\ev} \;\; \text{(by definition of $\sem{\te}$)} \\
& \implies^2 &  (\ev\catop\evtr,\subs)\in\sem{\te_2} \wedge (\emptyevtr,\emptyset)\in\sem{\te_1} \wedge (\emptyevtr\catop\ev)\nopref\sem{\te_1} \;\; \text{(by Lemma~\ref{lemma5})} \\
& \implies^2 & (\ev\catop\evtr,\subs)\in\sem{\te_1}\ccatop\sem{\te_2} \;\; \text{(by definition of $\ccatop$)} \\
 \end{eqnarray*}
 
 We now prove \emph{soundness}. To prove it, we show that the composition of abstract semantics $\sem{\te_1}$ and $\sem{\te_2}$ using the $\ccatop$ operator is included in the abstract semantics of the related concatenation term $\sem{\te_1\catop\te_2}$. The resulting proof is splitted in four separated cases. When the trace belongs to $\sem{\te_1}$ is infinite ($\implies^1$). The proof is based on the fact that an infinite trace concatenated to another trace is always equal to itself. In all the other cases, the proof can be  fully derived by a direct application of the operational semantics. 

 \begin{eqnarray*}
(\ev\catop\evtr,\subs)\in\sem{\te_1}\ccatop\sem{\te_2} & \implies & (\ev\catop\evtr)\in\inft{\sem{\te_1}} \vee \\ 
& & (\ev\catop\evtr = \evtr_1\catop\evtr_2 \wedge (\evtr_1,\subs_1)\in\sem{\te_1} \wedge (\evtr_2,\subs_2)\in\sem{\te_2} \wedge \subs = \subs_1\subsMerge\subs_2 \wedge \\
& & (\evtr_2 = \emptyevtr \vee (\evtr_2 = \ev'\catop\evtr_3 \wedge \evtr_1\catop\ev'\nopref\sem{\te_1}))) \;\; \text{(by definition of $\ccatop$)} \\
& \implies & (\ev\catop\evtr)\in\inft{\sem{\te_1}} \vee \\ 
& & (\evtr_1=\emptyevtr \wedge (\emptyevtr,\emptyset)\in\sem{\te_1} \wedge (\ev\catop\evtr,\subs)\in\sem{\te_2} \wedge \ev\nopref\sem{\te_1}) \vee \\
& & (\evtr_2 = \emptyevtr \wedge (\ev\catop\evtr)\in\sem{\te_1} \wedge (\emptyevtr,\emptyset)\in\sem{\te_2}) \vee \\
& & (\evtr_1=\ev\catop\evtr_1' \wedge \evtr_2=\ev'\catop\evtr_3 \wedge \evtr_1\catop\ev'\nopref{}\sem{\te_1} \wedge \\
& & (\ev\catop\evtr_1,\subs_1)\in\sem{\te_1} \wedge (\ev'\catop\evtr_3)\in\sem{\te_2} \wedge \subs=\subs_1\subsMerge\subs_2) \\
\end{eqnarray*}
\vspace*{-1.3cm}
\begin{eqnarray*}
(\ev\catop\evtr,\subs)\in\inft{\sem{\te_1}} & \implies^1 & (\ev\catop\evtr,\subs)\in\sem{\te_1} \wedge \evtr \text{ infinite} \;\; \text{(by definition of $\inft{}$)} \\
& \implies^1 & (\ev\catop\evtr,\substr)\in\csem{\te_1} \wedge (\ev\catop\evtr,\substr)\leadsto(\ev\catop\evtr,\subs) \wedge \\
& &  \evtr \text{ infinite} \;\; \text{(by definition of $\sem{\te}$)} \\
& \implies^1 & \te_1\extrans{\ev}\te_1';\subs' \text{ is derivable } \wedge (\evtr,\substr')\in\csem{\subs'\te_1'} \wedge \\
& & (\ev\catop\evtr,\substr)\leadsto(\ev\catop\evtr,\subs) \wedge \evtr \text{ infinite} \;\; \text{(by definition of $\csem{\te}$)} \\
& \implies^1 & \te_1\extrans{\ev}\te_1';\subs' \text{ is derivable } \wedge (\evtr,\substr')\in\csem{\subs'(\te_1'\catop\te_2)} \wedge \\
& & (\ev\catop\evtr,\substr)\leadsto(\ev\catop\evtr,\subs) \wedge \evtr \text{ infinite} \;\; \text{(by Lemma~\ref{lemmaOmega})} \\
& \implies^1 & \te_1\catop\te_2\extrans{\ev}\te_1'\catop\te_2;\subs' \text{ is derivable } \wedge (\evtr,\substr')\in\csem{\subs'(\te_1'\catop\te_2)} \wedge \\
& & (\ev\catop\evtr,\substr)\leadsto(\ev\catop\evtr,\subs) \wedge \evtr \text{ infinite} \;\; \text{(by operational semantics)} \\
& \implies^1 & (\ev\catop\evtr,\substr)\in\csem{\te_1\catop\te_2} \wedge (\ev\catop\evtr,\substr)\leadsto(\ev\catop\evtr,\subs) \;\; \text{(by definition of $\csem{\te}$)} \\
& \implies^1 & (\ev\catop\evtr,\subs)\in\sem{\te_1\catop\te_2} \;\; \text{(by definition of $\sem{\te}$)} \\
\end{eqnarray*}
\vspace*{-1cm}
\begin{eqnarray*}
(\evtr_1=\emptyevtr \wedge (\emptyevtr,\emptyset)\in\sem{\te_1} \wedge & & \\
(\ev\catop\evtr,\subs)\in\sem{\te_2} \wedge \ev\nopref\sem{\te_1}) & \implies^2 & \isEmpty(\te_1) \text{ is derivable } \wedge (\ev\catop\evtr,\subs)\in\sem{\te_2} \wedge \\
& & \ev\nopref\sem{\te_1} \;\; \text{(by definition of $\sem{\te}$)} \\
& \implies^2 & \isEmpty(\te_1) \text{ is derivable } \wedge \te_2\trans{\ev}\te_2';\subs' \text{ is derivable } \wedge \\
& & (\evtr,\substr')\in\csem{\subs'\te_2'} \wedge \ev\nopref\sem{\te_1} \wedge (\ev\catop\evtr,\substr)\leadsto(\ev\catop\evtr,\subs) \\
& & \text{(by definition of $\csem{\te}$)} \\
& \implies^2 & \te_1\catop\te_2\trans{\ev}\te_2';\subs' \text{ is derivable } \wedge (\evtr,\substr')\in\csem{\subs'\te_2'} \\
& & \text{(by operational semantics)} \\
& \implies^2 & (\ev\catop\evtr,\substr)\in\csem{\te_1\catop\te_2} \wedge (\ev\catop\evtr,\substr)\leadsto(\ev\catop\evtr,\subs) \\
& & \text{(by definition of $\csem{\te}$)} \\
& \implies^2 & (\ev\catop\evtr,\subs)\in\sem{\te_1\catop\te_2} \;\; \text{(by definition of $\sem{\te}$)} \\
 \end{eqnarray*}
\vspace*{-1.2cm}
\begin{eqnarray*}
(\evtr_2 = \emptyevtr \wedge (\ev\catop\evtr)\in\sem{\te_1} \wedge (\emptyevtr,\emptyset)\in\sem{\te_2}) & \implies^3 & (\ev\catop\evtr,\subs)\in\sem{\te_1\catop\te_2}  \;\; \text{(by Lemma~\ref{lemmaEmptyTail})}\\
\end{eqnarray*}
\vspace*{-1.2cm}
\begin{eqnarray*}
(\evtr_1=\ev\catop\evtr_1' \wedge \evtr_2=\ev'\catop\evtr_3 \wedge & &  \\
\evtr_1\catop\ev'\nopref{}\sem{\te_1} \wedge (\ev\catop\evtr_1,\subs_1)\in\sem{\te_1} \wedge & & \\ 
(\ev'\catop\evtr_3,\subs_2)\in\sem{\te_2} \wedge \subs=\subs_1\subsMerge\subs_2) & \implies^4 & \te_1\extrans{\ev}\te_1';\subs_1' \text{ is derivable } \wedge (\evtr_1,\substr_1)\in\csem{\subs_1'\te_1} \wedge \\
& & (\evtr_1,\substr_1)\leadsto(\evtr_1,\subs_1'') \wedge \te_2\trans{\ev'}\te_2';\subs_2' \text{ is derivable } \wedge \\
& & (\evtr_3,\substr_2')\in\sem{\subs_2'\te_2'} \wedge (\evtr_2,\substr_2')\leadsto(\evtr_2,\subs_2'') \wedge \te_1\notrans{\ev'} \text{ is derivable} \\
& & \wedge \subs_1=\subs_1'\subsMerge\subs_1'' \wedge \subs_2=\subs_2'\subsMerge\subs_2'' \;\; \text{(by operational semantics)} \\
& \implies^4 & (\evtr_1\catop\evtr_2,\substr)\in\csem{\te_1\catop\te_2} \;\; \text{(by Lemma~\ref{lemmaConcatExpansion})} \\
\end{eqnarray*}
\vspace*{-1cm}

%% file: related.tex
\added{
\section{Related Work}
\label{sect:related_work}

Compositionality, determinism and events-based semantics are topics
very central to concurrent systems.  Winskel has introduced the notion
of event structure \cite{Winskel86} to model computational processes
as sets of event occurrences together with relations representing
their causal dependencies.  Vaandrager \cite{Vaandrager91} has proved
that for concurrent deterministic systems it is sufficient to observe
the beginning and end of events to derive its causal structure.  Lynch
and Tuttle have introduced input/output automata \cite{LynchT87} to
model concurrent and distributed discrete event systems with a trace
semantics consisting of both finite and infinite sequences of actions.

The rest of this section describes some of the main RV techniques and state-of-the-art tools and compares them with respect to \rml;
more comprehensive surveys on RV can be found in literature \cite{DelgadoGR04,HavelundG05,rv,SokolskyHL12,Falcone13,BartocciFFR18,HavelundRTZ18} which mention formalisms for parameterised runtime verification that have not deliberately presented here for space limitation.
\paragraph{Monitor-oriented programming:}
Similarly as \rml, which does not depend on the monitored system and its instrumentation, other proposals introduce different levels of separation of concerns.
\emph{Monitor-oriented programming} (\emph{MOP} \cite{ChenR07}) is an infrastructure for RV that is neither tied to any particular programming language nor to a single specification language. 
In order to add support for new logics, one has to develop an appropriate plug-in converting specifications to one of the format supported by the MOP instance of the language of choice; 
the main formalisms implemented in existing MOP include finite state machines, extended regular expressions, context-free grammars and
temporal logics.
Finite state machines (or, equivalently, regular expressions) can be easily translated to \rml, have limited expressiveness,
but are widely used in RV  because they are well-understood among software developers
as opposite to other more sophisticated approaches, as temporal logics.
Extended regular expressions include intersection and complement; although such
operators allow users to write more compact specifications, they do not increase the formal expressive power since
regular languages are closed under both.
Deterministic Context-Free grammars (that is, deterministic pushdown automata) can be translated in \rml using recursion, concatenation, union, and the empty trace, while the relationship with Context-Free grammars (that is, pushdown automata) has not been fully investigated yet;
as stated in the introduction, \rml can express several non Context-Free properties, hence \rml cannot be less expressive
than Context-Free grammars, but we do not know whether Context-Free grammars are less expressive than \rml.



\paragraph{Temporal logics:}
Since RV has its roots in model checking, it is not surprising that logic-based formalism previously introduced in the context of the latter have been applied to the former.
\emph{Linear Temporal Logic} (\emph{LTL}) \cite{ltl}, is one of the most used formalism in verification.

Since the standard semantics of LTL is defined on infinite traces only, and RV monitors can only
check finite trace prefixes (as opposed to static formal verification), a three-valued semantics for LTL, named \emph{LTL\textsubscript{3}} has been proposed \cite{ltl3}.
Beyond the basic ``true'' and ``false'' truth values, a third ``inconclusive'' one is considered (LTL specification syntax is unchanged, only the semantics is modified to take into account the new value).
This allows one to distinguish the satisfaction/violation of the desired property (``false'') from the lack of sufficient evidence among the events observed so far (``inconclusive''), making this semantics more suited to RV.
Differently from LTL, the semantics of LTL\textsubscript{3} is defined on finite prefixes, making it more suitable for comparison with other RV formalisms.
Further development of LTL\textsubscript{3} led to \emph{RV-LTL} \cite{ltl4}, a 4-valued semantics on which \rml monitor verdicts are based on.

The expressive power of LTL is the same as of star-free $\omega$-regular languages \cite{PnueliZ93}.
When restricted to finite traces, \rml is much more expressive than LTL as any regular expression can be trivially translated to it; however,
on infinite traces, the comparison is slightly more intricate since \rml and LTL\textsubscript{3} have incomparable expressiveness \cite{AnconaFM16}.
There exist many extensions of LTL that deal with time in a more quantitative way (as opposed to the strictly qualitative approach of standard LTL) without increasing the expressive power, like \emph{interval temporal logic} \cite{CauZ97}, \emph{metric temporal logic} \cite{ThatiR05} and \emph{timed LTL} \cite{ltl3}.
Other proposals go beyond regularity \cite{AlurEM04} and even context-free languages \cite{BolligDL12}.

Several temporal logics are embeddable in \emph{recHML} \cite{Larsen90}, a variant of the \emph{modal $\mu$-calculus} \cite{Kozen83};
this allows the formal study of monitorability \cite{AcetoEtAl2019} in a general framework, to derive results for free about any formalism that can be
expressed in such calculi.
It would be interesting to study whether the \rml trace calculus could be derivable to get theoretical results that are missing from this presentation.
Unfortunately, it is not clear whether our calculus and \emph{recHML} are comparable at all.
For instance, \emph{recHML} is a fixed-point logic including both least and greatest fixpoint operators, while our calculus implicitly uses a greatest fixpoint semantics for recursion. Nonetheless, \emph{recHML} does not include a shuffle operator, and we are not aware of a way to derive it from other operators.

Regardless of the formal expressiveness, \rml and temporal logics are essentially different: \rml is closer to formalisms
with which software developers are more familiar, as regular expressions and
Context-Free languages, but does not offer direct support for time;
however, if the instrumentation provides timestamps, then some time-related properties can still be expressed exploiting parametricity.

\paragraph{State machines:}

As opposite to the language-based approach, as \rml, specifications can be defined using \emph{state machines} (a.k.a.\ automata or finite-state machines).
Though the core concept of a finite set of states and a (possibly input-driven) transition function between them is always there, in the field of automata theory different formalizations and extensions bring the expressiveness anywhere from simple deterministic finite automata to Turing machines.

An example of such formalisms is \emph{DATE} (Dynamic Automata with Timers and Events \cite{DATEs}), an extension of the finite-state automata computational model based on communicating automata with timers and transitions triggered by observed events.
This is the basis of \emph{LARVA} 
\cite{ColomboPS09}, a Java RV tool focused on control-flow and real-time properties, exploiting the expressiveness of the underlying system (DATE).

The main feature of LARVA that is missing in \rml is the support for temporized properties, as observed events can trigger timers for other expected events.
On the other hand, the parametric verification support of \rml is more general.
LARVA scope mechanism works at the object level, thus parametricity is based on \emph{trace slicing} \cite{HavelundRTZ18} and implemented by spawning new monitors and associating them with different objects.
The \rml approach is different as specifications can be parametric with respect to any observed data thanks to 
event type patterns and the let-construct to control the scope of the variables occurring in them.
Limitations of the parametric trace slicing approach described above, as well as possible generalizations to overcome them, have been explored by \cite{ChenR09,BarringerFHRR12,RegerCR15}.

Finally, the goals of the two tools are different:
while \rml strives to be system-independent, LARVA is devoted to Java verification, and the implementation relies on AspectJ 
\cite{KiczalesHHKPG01} as an ``instrumentation'' layer allowing one to inject code (the monitor) to be executed at specific locations in the program.
}

%% file: conclu.tex
\section{Conclusion}\label{sec:conclu}

We have moved a first step towards a compositional semantics of the \rml trace calculus, by introducing
the notion of instantiated event trace, defining the semantics of trace expressions in terms of sets of
instantiated event traces and showing how each basic trace expression operator can be interpreted as
an operation over sets of instantiated event traces; we have formally proved that such an interpretation is
equivalent to the semantics derived from the transition system of the calculus if one considers only
contractive terms.

For simplicity, here we have considered only the core of the calculus, but we plan to extend our result to
the full calculus, which includes also the prefix closure operator and a top-level layer with constructs to support
generic specifications \cite{FranceschiniPhD2020}. Another interesting direction for further investigation consists in studying how the notion
of contractivity influences the expressive power of the calculus and, hence, of \rml; although we have failed so far to find a non-contractive
term whose semantics is not equivalent to a corresponding contractive trace expression, we have not formally proved that contractivity does not limit the
expressive power of the calculus.